\documentclass[10pt,journal,compsoc]{IEEEtran}
\usepackage[boxed,ruled,vlined,linesnumbered]{algorithm2e}
\usepackage{multirow}
\usepackage{hhline}
\usepackage{url}
\usepackage{framed}
\usepackage{longtable}
\usepackage{mdframed}

\newcommand*\wrapletters[1]{\wr@pletters#1\@nil}
\def\wr@pletters#1#2\@nil{#1\allowbreak\if&#2&\else\wr@pletters#2\@nil\fi}
\usepackage{tikz}
\usetikzlibrary{automata,matrix,shapes,arrows,positioning,chains,calc}
\usetikzlibrary{snakes}
\usetikzlibrary{arrows,scopes}
\usetikzlibrary{positioning,chains,fit,shapes,calc}
\usepackage{caption}
\usepackage{subcaption}
\usepackage{amsmath}
\SetKwHangingKw{Local}{Local}
\SetKwHangingKw{Global}{Global}
\SetKwHangingKw{Constant}{Constant}
\SetKwHangingKw{Input}{Input}
\SetKwHangingKw{Alias}{Alias}
\SetKwHangingKw{Output}{Output}
\SetKwHangingKw{SideEffect}{Side effects}
\SetKwHangingKw{InternalEvent}{Internal event}
\SetKwHangingKw{External}{External}
\SetKwHangingKw{ExternalEvent}{External event}
\SetKwHangingKw{Precondition}{Precondition :}
\SetKwHangingKw{Postcondition}{Postcondition :}
\SetKwHangingKw{Upon}{Upon}
\SetKwHangingKw{DoForever}{Do Forever}
\SetKwHangingKw{PVab}{Persistent variables:}

\usepackage{theorem}
\usepackage{wrapfig}
\usepackage{amssymb}
\usepackage{xspace}

\newtheorem{theorem}{Theorem}

\newcommand{\qed}{\hfill $\blacksquare$}
\newcommand{\qedClaim}{\hfill \ensuremath{\Box}}
\newenvironment{proof}{\noindent {\bf Proof.}\ }{\qed\par\vskip 0mm\par}

\newcommand{\B}{\vspace*{-\smallskipamount}}

\newcommand{\BBB}{\vspace*{-\bigskipamount}}

\newcommand{\remove}[1]{}

\usepackage{csquotes}
\usepackage{enumitem}
\usepackage{tablefootnote}

\ifCLASSOPTIONcompsoc
  \usepackage[nocompress]{cite}
\else
  \usepackage{cite}
\fi
\ifCLASSINFOpdf
\else
\fi

\IEEEoverridecommandlockouts

\begin{document}
\title{Meta-MapReduce\\A Technique for Reducing Communication in MapReduce Computations\thanks{Foto Afrati is withNational Technical University of Athens, Greece (e-mail: \texttt{afrati@softlab.ece.ntua.gr}).\protect\\
Shlomi Dolev is with Ben-Gurion University of the Negev, Beer-Sheva, Israel (e-mail: \texttt{dolev@cs.bgu.ac.il}). \protect\\
Shantanu Sharma is with Ben-Gurion University of the Negev, Beer-Sheva, Israel (e-mail: \texttt{sharmas@cs.bgu.ac.il}). \protect\\
Jeffrey D. Ullman is with Stanford University, USA (e-mail: {ullman@cs.stanford.edu}).\protect\\
A brief announcement this work is accepted in 17th International Symposium on Stabilization, Safety, and Security of Distributed Systems (SSS), 2015. This work of the first author is supported by the project Handling Uncertainty in Data Intensive Applications, co-financed by the European Union (European Social Fund) and Greek national funds, through the Operational Program \enquote{Education and Lifelong Learning,} under the program THALES. This work of the second author is supported by the Rita Altura Trust Chair in Computer Sciences, Lynne and William Frankel Center for Computer Sciences, Israel Science Foundation (grant 428/11), the Israeli Internet Association, and the Ministry of Science and Technology, Infrastructure Research in the Field of Advanced Computing and Cyber Security.}}

%\author{Shlomi Dolev,~\IEEEmembership{Senior Member,~IEEE}, Yin Li, Shantanu Sharma,~\IEEEmembership{Student Member,~IEEE}}

\author{Foto Afrati, Shlomi Dolev, Shantanu Sharma, and Jeffrey D. Ullman}

%% The paper headers
%\markboth{IEEE Transaction on Cloud Computing}%
%{Shell \MakeLowercase{\textit{et al.}}: Bare Advanced Demo of IEEEtran.cls for Journals}
\IEEEtitleabstractindextext{
\begin{abstract}
MapReduce has proven to be one of the most useful paradigms in the revolution of distributed computing, where cloud services and cluster computing become the standard venue for computing. The federation of cloud and big data activities is the next challenge where MapReduce should be modified to avoid (big) data migration across remote (cloud) sites. This is exactly our scope of research, where only the very essential data for obtaining the result is transmitted, reducing communication, processing and preserving data privacy as much as possible. In this work, we propose an algorithmic technique for MapReduce algorithms, called Meta-MapReduce, that decreases the communication cost by allowing us to process and move metadata to clouds and from the map phase to reduce phase. In Meta-MapReduce, the reduce phase fetches only the required data at required iterations, which in turn, assists in preserving the data privacy.
\end{abstract}
}
\maketitle

\IEEEdisplaynontitleabstractindextext
\IEEEpeerreviewmaketitle
\ifCLASSOPTIONcompsoc
\IEEEraisesectionheading{\section{Introduction}\label{sec:introduction}}
\else

\section{Introduction}
\label{section:introduction}
\fi
MapReduce~\cite{DBLP:conf/osdi/DeanG04} is a programming system used for parallel processing of large-scale data. The given input data is processed by the \emph{map phase} that applies a user-defined map function to each input and produces intermediate data (of the form $\langle \mathit{key, value} \rangle$). Afterwards, intermediate data is processed by the \emph{reduce phase} that applies a user-defined reduce function to $\mathit{key}$s and all their associated $\mathit{value}$s. The final output is provided by the reduce phase. A detailed description of MapReduce can be found in Chapter 2 of~\cite{DBLP:books/ullman2011}.

\parskip 0pt
\setlength{\parindent}{15pt}

%\medskip
\noindent \textit{Mappers, Reducers, and Reducer Capacity.} A \emph{mapper} is an application of a map function to a single input. A \emph{reducer} is an application of a reduce function to a single $\mathit{key}$ and its associated list of $\mathit{value}$s. The \emph{reducer capacity}~\cite{DBLP:journals/corr/AfratiDK0U15a} --- an important parameter --- is an upper bound on the sum of the sizes of the $\mathit{value}$s that are sent to the reducer. For example, the reducer capacity may be the size of the main memory of the processors on which the reducers run. The capacity of a reducer is denoted by $q$, and all the reducers have an identical capacity.

%\medskip\noindent\textbf{Locality of data and communication cost.}
\subsection{Locality of Data and Communication Cost}
Input data to a MapReduce job, on one hand, may exist at the same site where mappers and reducers reside. However, ensuring an identical location of data and mappers-reducers cannot always be guaranteed. On the other hand, it may be possible that a user has a single local machine and wants to enlist a public cloud to help data processing. Consequently, in both the cases, it is required to move data to the location of mappers-reducers. Interested readers may see examples of MapReduce computations where the locations of data and mappers-reducers are different in~\cite{heintzend,ghadoop}. A review on geographically distributed data processing frameworks based on MapReduce may be found in~\cite{geo-hadoop}.

\parskip 0pt
\setlength{\parindent}{15pt}

In order to motivate and demonstrate the impact of different locations of data and mappers-reducers, we consider two real examples, as follows:
\begin{description}[nolistsep]
  \item[Amazon Elastic MapReduce] Amazon Elastic MapReduce (EMR)\footnote{http://aws.amazon.com/elasticmapreduce/} processes data that is stored in Amazon Simple Storage Service (S3)\footnote{http://aws.amazon.com/s3/}, where the locations of EMR and S3 are not identical. Hence, it is required to move data from S3 to the location of EMR. However, moving the whole dataset from S3 to EMR is not efficient if only small specific part of it is needed for the final output.
  \item[G-Hadoop and Hierarchical MapReduce] Two new implementations of MapReduce, G-Hadoop~\cite{ghadoop} and Hierarchical MapReduce~\cite{hmr}, perform MapReduce computations over geographically distributed clusters. In these new implementations, several clusters execute an assigned MapReduce job in parallel, and provide \emph{partial} outputs. Note that the output of a cluster is not the final output of a MapReduce job, and the final output is produced by processing partial outputs of all the clusters at a single cluster. Thus, inter-cluster data transfer --- transmission of partial outputs of all the clusters to a single cluster --- is required for producing the final output, as the location of the partial outputs of all the clusters and the location of the final computation are not identical. However, moving the whole partial outputs of all the clusters to a single cluster is also not efficient if only small portion of the clusters' outputs is needed to compute the final output.

\begin{figure}[h]
\centering
\includegraphics[scale=0.32]{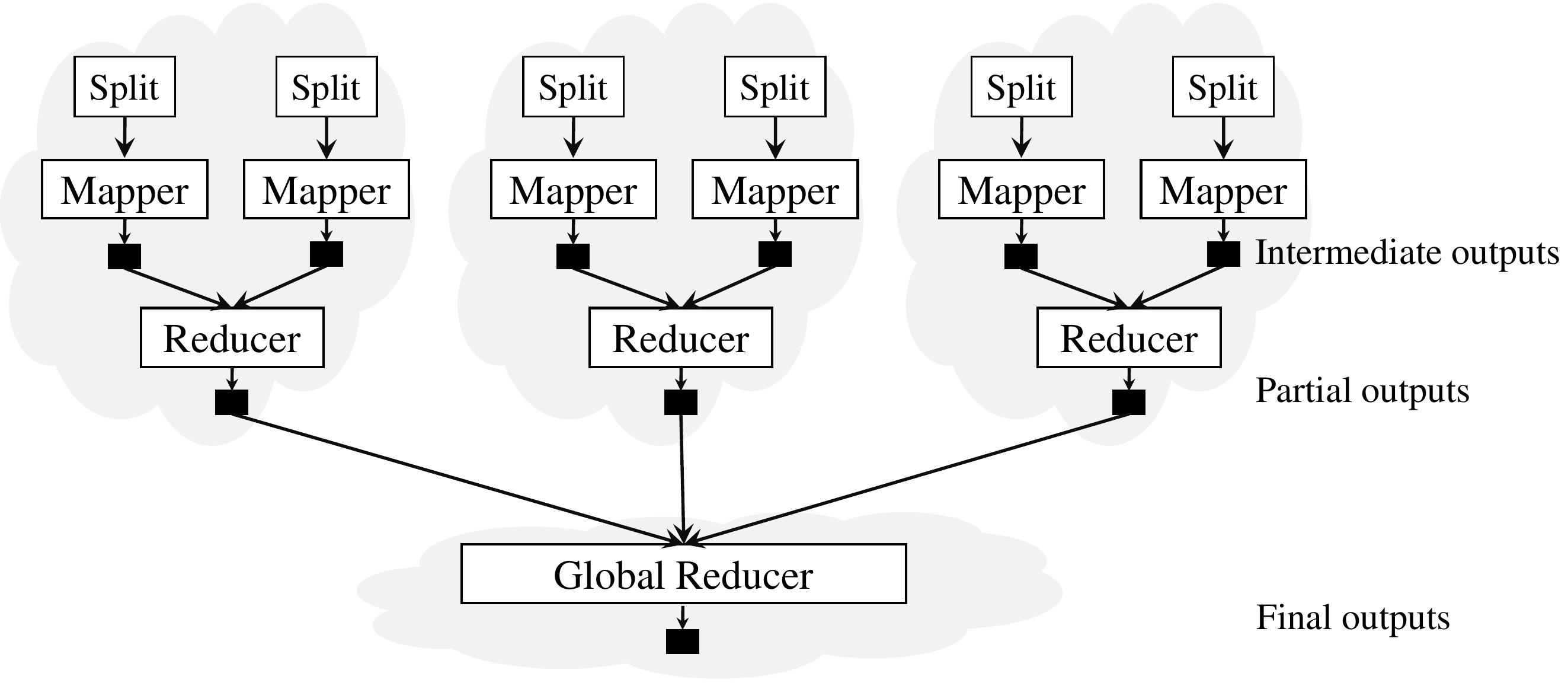}
\caption{Different locations of MapReduce clusters in Hierarchical MapReduce.}
\label{fig:pic_hmr}
\end{figure}

Hierarchical MapReduce is depicted in Figure~\ref{fig:pic_hmr}, where three clusters process data using MapReduce in parallel, and the output of all three clusters is required to be sent to one of the clusters or another cluster, which executes a \emph{global reducer} for providing the final output. In Figure~\ref{fig:pic_hmr}, it is clear that the locations of partial outputs of all the clusters and the location of the global reducers is not identical; hence, partial outputs of all the clusters are required to be transferred to the location of the global reducer.
\end{description}

\medskip\noindent\textbf{Communication cost}. The \textit{communication cost} dominates the performance of a MapReduce algorithm and is the sum of the amount of data that is required to move from the location of users or data (\textit{e}.\textit{g}., S3) to the location of mappers (\textit{e}.\textit{g}., EMR) and from the map phase to the reduce phase in each round of a MapReduce job. For example, in Figure~\ref{fig:pic_hmr}, the communication cost is the sum of the total amount of data that is transferred from mappers to reducers in each cluster and from each cluster to the site of the global reducer.

In this paper, we are interested in minimizing the data transferred in order to avoid communication and memory overhead, as well as to protect data privacy as much as possible. In MapReduce, we transfer inputs to the site of mappers-reducers from the site of the user, and then, several copies of inputs from the map phase to the reduce phase in each iteration, regardless of their involvement in the final output. If few inputs are required to compute the final output, then it is not communication efficient to move all the inputs to the site of mappers-reducers, and then, the copies of same inputs to the reduce phase.

There are some works that consider the location of data~\cite{Palanisamy-Purlieus-2011,DBLP:conf/hpdc/ParkLKHM12} in a restrictive manner and some works~\cite{DBLP:journals/corr/abs-1206-4377,ding2013commapreduce,TCS} that consider data movement from the map phase to the reduce phase. We enhance the model suggested in~\cite{TCS} and suggest an algorithmic technique for MapReduce algorithms to decrease the communication cost by moving only relevant input data to the site of mappers-reducers. Specifically, we move metadata of each input instead of the actual data, execute MapReduce algorithms on the metadata, and only then fetch the actual required data needed to compute the final output.

%\subsection{Motivating Examples}

%\medskip\medskip \noindent \textbf{Motivating examples.}
\subsection{Motivating Examples}
We present two examples (equijoin and entity resolution) to show the impact of different locations of data and mappers-reducers on the communication cost involved in a MapReduce job.

\parskip 0pt
\setlength{\parindent}{15pt}

\medskip
\noindent\textbf{Equijoin of two relations $X(A,B)$ and $Y(B,C)$.} \textit{Problem statement}: The join of relations $X(A,B)$ and $Y(B,C)$, where the joining attribute is $B$, provides output tuples $\langle a, b, c\rangle$, where $(a,b)$ is in $A$ and $(b,c)$ is in $C$. In the equijoin of $X(A,B)$ and $Y(B,C)$, all tuples of both the relations with an identical value of the attribute $B$ should appear together at the same reducer for providing the final output tuples. In short, a mapper takes a single tuple from $X(A,B)$ or $Y(B,C)$ and provides $\langle B,X(A)\rangle$ or $\langle B,Y(C)\rangle$ as key-value pairs. A reducer joins the assigned tuples who have an identical key. In Figure~\ref{fig:join_mr}, two relations $X(A,B)$ and $Y(B,C)$ are shown, and \textit{we consider that the size of all the $B$ values is very small as compared to the size of values of the attributes $A$ and $C$}.

\begin{figure}[h]
\begin{center}
\includegraphics[scale=0.35]{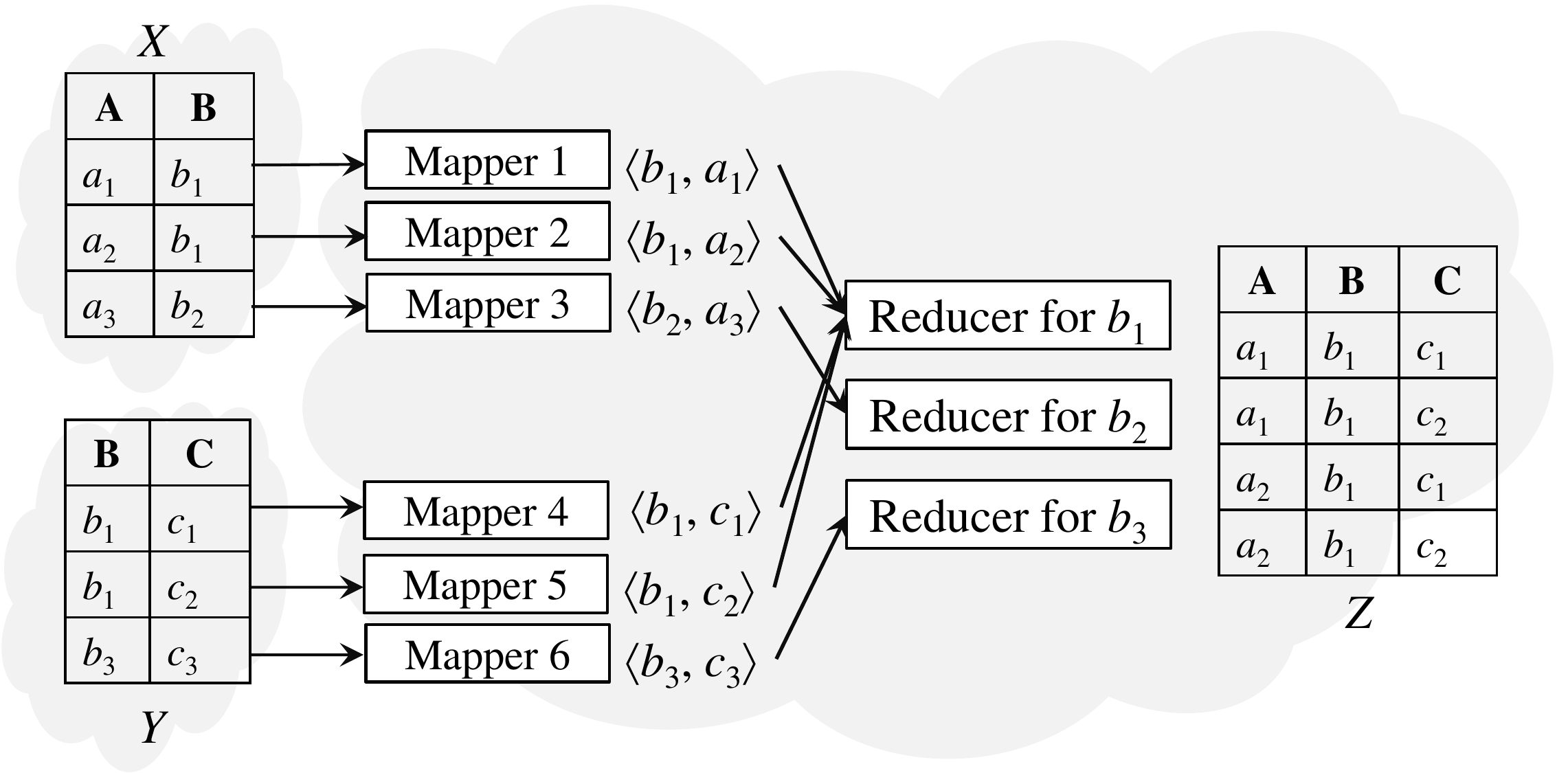}
\end{center}
\caption{Equijoin of relations $X(A,B)$ and $Y(B,C)$.}
\label{fig:join_mr}
\end{figure}

\textit{Communication cost analysis}: We now investigate the impact of different locations of the relations and mappers-reducers on the communication cost. In Figure~\ref{fig:join_mr}, the communication cost for joining of the relations $X$ and $Y$ --- where $X$ and $Y$ are located at two different clouds and equijoin is performed on a third cloud --- is the sum of the sizes of all three tuples of each relation that are required to move from the location of the user to the location of mappers, and then, from the map phase to the reduce phase. Consider that each tuple is of unit size, and hence, the total communication cost is 12 for obtaining the final output.

However, if there are a few tuples having an identical $B$-value in both the relations, then it is useless to move the whole relations from the user's location to the location of mappers, and then, tuples from the map phase to the reduce phase. In Figure~\ref{fig:join_mr}, two tuples of $X$ and two tuples of $Y$ have a common $B$ value (\textit{i}.\textit{e}., $b_1$). Hence, it is not efficient to send tuples having values $b_2$ and $b_3$, and by not sending tuples with $B$ values $b_2$ and $b_3$, we can reduce the communication cost.

\medskip\noindent \textbf{Entity resolution.} \textit{Problem statement}: Entity resolution using MapReduce is suggested in~\cite{TCS}. A solution to the entity resolution problem provides disjoint subsets of records, where records match other records if they pass a similarity function, and these records belong to a single entity (or person). For example, if voter-card, passport, student-id, driving-license, and phone numbers of several people are given, a solution to the entity resolution problem makes several subsets, where a subset corresponding to a person holds their voter-card, passport, student-id, driving-license, and phone numbers.

\textit{Communication cost analysis}: The authors~\cite{TCS} provided a solution to the entity resolution problem with decreased the communication cost between the map phase and the reduce phase, by not sending the record of a person to a reducer, if the person has only a single record. However, in this model for every pair of records that share a reducer, a copy of one is transferred to the site of reducers, and in this manner the communication cost is $\frac{n(n-1)}{2}$, if $n$ records are assigned to a reducer. Note that by using the proposed technique of this paper, the communication cost for the entity resolution problem will be $n$, if $n$ records are assigned to a reducer.

%Other examples that fit in the four parameters are the drug-interaction problem~\cite{DBLP:journals/crossroads/Ullman12} (where a set of inputs consists of 6,500 drugs, and a drug $i$ holds information about the medical history of patients who had taken drug $i$. The objective is to find pairs of drugs that have particular side effects), computing common friends, different kinds of joins (interval join, spatial join, and star join), outer product, and Kronecker product.

\subsection{Problem Statement and Our Contribution}
%\medskip\noindent \textbf{Problem statement.}
We are interested in reducing the amount of data to be transferred to the site of the cloud executing MapReduce computations and the amount of data transferred from the map phase to the reduce phase. From the preceding two examples, it is clear that in several problems, the final output depends on \textit{some} inputs, and in those cases, it is not required to send the whole input data to the site of mappers and then (intermediate output) data from the map phase to the reduce phase. Specifically, we consider two scenarios for reducing the communication cost, where: (\textit{i}) the locations of data and mappers-reducers are different, and (\textit{ii}) the locations of data and mappers are identical. Note that in the first case, data is required to move from the user's site to the computation site and then from the map phase to the reducer phase, while in the second case, data is transferred only from the map phase to the reduce phase.

\parskip 0pt
\setlength{\parindent}{15pt}

In addition to the locations of data and computations, we are also considering the number of iterations involved in a MapReduce job and the reducer capacity (\textit{i}.\textit{e}., the maximum size of inputs that can be assigned to a reducer).

\medskip \noindent \textbf{Our contributions.} In this paper, we provide the following:
\begin{itemize}[nolistsep,leftmargin=0.5cm]
%\item Are considering, for the first time, four important parameters (the communication cost, the locality of data and mappers-reducers, the reducer capacity, and the total number of iterations) of a MapReduce job all together. Specifically, we focus on two scenarios where the locations of data and mappers-reducers are identical and non-identical. Keeping in mind these two scenarios, we try to decrease the total communication cost in a MapReduce job.

\item \textbf{An algorithmic approach for MapReduce algorithms.} We provide a new \textit{algorithmic} approach for MapReduce algorithms, Meta-MapReduce (Section~\ref{section:meta-mapreduce}), that decreases the communication cost significantly. Meta-MapReduce regards the locality of data and mappers-reducers and avoids the movement of data that does not participate in the final output. Particularly, \textbf{Meta-MapReduce provides a way to compute the desired output using metadata}\footnote{The term metadata is used in a different manner, and it represents a small subset, which varies according to tasks, of the dataset.} (which is much smaller than the original input data) and avoids to upload the whole data (either because it takes too long or for privacy reasons). It should be noted that we are enhancing MapReduce and not creating entirely a new framework for large-scale data processing; thus, Meta-MapReduce is \textit{implementable} in the state-of-the art MapReduce systems such as Spark~\cite{zaharia2010spark}, Pregel~\cite{DBLP:conf/sigmod/MalewiczABDHLC10}, or modern Hadoop.

\item \textbf{Data-privacy in MapReduce computations.} Meta-MapReduce also allows us to protect data privacy as much as possible in the case of an \emph{honest-but-curious} adversary by not sending all the inputs. For example, in the case of equijoin, processing tuple $\langle a_i, b_i \rangle$ of a relation $X$ and $\langle b_j, c_i \rangle$ of a relation $Y$ based on metadata does not reveal the actual tuple information until it is required at the cloud. It should be noted that the relations $X$ and $Y$ can deduce that the relations $Y$ and $X$ have no tuple with value $b_i$ and $b_j$, respectively. However, the outcome of both relations do not imply the actual value of $a_i$, $b_i$, and $c_i$.

    Nevertheless, by the auditing process, a \emph{malicious} adversary can be detected. Moreover, in some settings auditing enforces participants to be honest-but-curious rather than malicious, as malicious actions can be audited, discovered, and imply punishing actions. %Note that our technique considers the existence of a reducer input cache, a reducer output cache, and a mapper input cache~\cite{DBLP:journals/vldb/BuHBE12}.

%\item Provide a tradeoff between the amount of data and the total number of read operations at the reduce phase (Sections~\ref{section:the_system_setting} and~\ref{subsec:Meta-MapReduce framework}) and a tradeoff between the amount of (meta)data and the total number of reducers (Section~\ref{subsec:Meta-MapReduce framework}).

%\item Provide a tradeoff between the amount of (meta)data and the total number of reducers (Section~\ref{subsec:Meta-MapReduce framework}).

\item \textbf{Other applications.}
\begin{itemize}[noitemsep,nolistsep,leftmargin=0.5cm]
\item Integrate Meta-MapReduce for processing geographically distributed data (Section~\ref{subsec:Incorporating Meta-MapReduce in G-Hadoop and Hierarchical MapReduce}) and for processing multi-round MapReduce jobs (Section~\ref{subsec:Multi-round Iterations}).

\item We show versatile applicability of Meta-MapReduce by performing equijoin, $k$-nearest-neighbors, and shortest path findings on a social networking graph using Meta-MapReduce and show how Meta-MapReduce decreases the communication cost (Sections~\ref{subsec:Meta-MapReduce framework} and~\ref{sec:Versatility of Meta-MapReduce}).
\end{itemize}
\end{itemize}

\subsection{Related Work}
%\medskip \noindent \textbf{Related work.}
MapReduce was introduced by Dean and Ghemawat in 2004~\cite{DBLP:conf/osdi/DeanG04}. ComMapReduce~\cite{ding2013commapreduce} decreases communication between the map and reduce phases (for some problems) by allowing mappers to inform a \textit{coordinator} about their outputs, and the coordinator finds those outputs of mappers that will not participate in the final output. The coordinator informs reducers not to fetch those (undesired) outputs of the map phase. However, the model does not allow computation over metadata for discovering the actual data needed. Thus, in their limited scope of data filtering, there is no need to consider, the number iterations, the reducer capacity, and the locality of data.

\parskip 0pt
\setlength{\parindent}{15pt}

Another framework, Purlieus~\cite{Palanisamy-Purlieus-2011}, considers the location of the data only and tries to deploy mappers and reducers near the location of data. In~\cite{TCS}, a model is given for computing records that belong to a single person. The model decreases the communication cost by not sending the record of a person to the reduce phase, if the person has only a single record. However, the model does not consider the reducer capacity and different locations of data and mappers-reducers. Afrati et al.~\cite{DBLP:journals/corr/AfratiDK0U15a} presents a model regarding the reducer capacity to find a lower bound on communication cost as a function of the total size of inputs sent to a reducer. The model considers a problem where an output depends on \emph{exactly} two inputs. However, the model assumes an identical location of data and mappers-reducers.

In summary, all the existing models provides a number of ways that MapReduce algorithms can be sped up if we are running on a platform that allows certain operations. However, these models~\cite{ding2013commapreduce,Palanisamy-Purlieus-2011,TCS,DBLP:journals/corr/AfratiDK0U15a} cannot be implemented for reducing the communication cost in the case of geographically distributed MapReduce~\cite{hmr,ghadoop} and for preserving data privacy. Moreover, all the models require data movement from the location of data to the location of computations and between the map and reduce phases irrespective of their involvement in outputs and the number of iterations of MapReduce jobs, resulting in a huge communication cost. In this paper, we explore hashing and filtering techniques for reducing the communication cost and try to move only relevant data to the site of mappers-reducers. The ability to use these algorithms depends on the capability of the platform, but many systems today, such as Pregel, Spark, or recent implementations of MapReduce offer the necessary capabilities.

\noindent \textbf{Bin-packing-based approximation algorithms.} We will use bin-packing-based approximation algorithms~\cite{DBLP:journals/corr/AfratiDK0U15a} for assigning inputs at reducers. These algorithms are designed for two classes of problems, which cover lots of MapReduce-based problems. The first problem deals with assigning each pair of inputs to at least one reducer in common, without exceeding $q$. The second problem deals with assign a pair of inputs to at least one reducer in common, without exceeding $q$, where two inputs belong to two different sets of inputs. The bin-packing-based algorithm works as follows:  use a known bin packing algorithm to place given $m$ inputs of size at most $\frac{q}{k}$ to bins of size $\frac{q}{k}$. Assume that we need $x$ bins to place $m$ inputs. Now, each of these bins is considered as a single input of size $\frac{q}{k}$ for our problem, and assign pair of these bins to reducers.

%%=====================================================================================================================================================
%%=====================================================================================================================================================
%% Section 2         Section 2       Section 2       Section 2       Section 2       Section 2       Section 2       Section 2       Section 2
%%=====================================================================================================================================================
%%=====================================================================================================================================================

\section{The System Setting}
\label{section:the_system_setting}
The system setting is an extension of the standard setting~\cite{DBLP:journals/corr/AfratiDK0U15a}, where we consider, for the first time, the locations of data and mappers-reducers and the communication cost. The setting is suitable for a variety of problems where \emph{at least} two inputs are required to produce an output. In order to produce an output, we need to define the term \emph{mapping schema}, as follows:

\parskip 0pt
\setlength{\parindent}{15pt}

\medskip\noindent \textbf{Mapping Schema.} A mapping schema is an assignment of the set of inputs (\textit{i}.\textit{e}., outputs of the map phase) to some given reducers under the following two constraints:

\begin{itemize}[noitemsep,nolistsep,leftmargin=0.5cm]
  \item A reducer is assigned inputs whose sum of the sizes is less than or equal to the reducer capacity $q$.
  \item For each output produced by reducers, we must assign the corresponding inputs to at least one reducer in common.
\end{itemize}
For example, a mapping schema for equijoin example will assign all tuples (of relations $X$ and $Y$) having an identical key to a reducer such that the size of assigned tuples is not more than $q$.

\medskip \noindent \textbf{The Model.} The model is simple but powerful and assumes the following:

\begin{enumerate}[noitemsep,nolistsep,leftmargin=0.5cm]
  \item Existence of systems such as Spark, Pregel, or modern Hadoop.

  \item A preliminary step at the user site who owns the dataset for finding metadata\footnote{The selection of metadata depends on the problem.} that has smaller memory size than the original data.

  \item Approximation algorithms (given in~\cite{DBLP:journals/corr/AfratiDK0U15a}), which are based on a bin-packing algorithm, at the cloud or the global reducer in case of Hierarchical MapReduce~\cite{hmr}. The approximation algorithms assign outputs of the map phase to reducers while regarding the reducer capacity. Particularly, in our case, approximation algorithms will assign metadata to reducers in such a manner that the size of actual data at a reducer will not exceed than the reducer capacity and all the inputs that are required to produce outputs must be assign at one reducer in common.
\end{enumerate}

%\medskip \noindent \textbf{Tradeoff.} There is a tradeoff between the amount of data (original data or metadata) and the total number of read operations at the reduce phase. (Consider that the transfer of either original data or metadata takes a unit time from the map phase to the reduce phase.) In the case of the small amount of data, \textit{i}.\textit{e}., metadata, reducers are required to read metadata and then (required) original input data, \textit{i}.\textit{e}., reducers do at least two-times read operations. On the other hand, in the case of original data, reducers are assigned the whole data, which increase the communication cost; however, reducers perform only one-time read operation.

\smallskip In the next section, we present a new algorithmic technique for MapReduce algorithms, where we try to minimize the communication cost regarding different locations of data and mappers-reducers with the help of a running example of equijoin.

\section{Meta-MapReduce}
\label{section:meta-mapreduce}
We present our algorithmic technique that reduces the communication cost for a variety of problems, \textit{e}.\textit{g}., join of relations, $k$-nearest-neighbors finding, similarity-search, and matrix multiplication. The proposed technique regards locality of data, the number of iterations involved in a MapReduce job, and the reducer capacity. The idea behind the proposed technique is to process metadata at mappers and reducers, and process the original required data at required iterations of a MapReduce job at reducers. In this manner, we suggest to process metadata at mappers and reducers at all the iterations of a MapReduce job. Therefore, the proposed technique is called \textit{Meta-MapReduce}.

\parskip 0pt
\setlength{\parindent}{15pt}

Before going into detail of Meta-MapReduce, we need to redefine the communication cost to takes into account the size of the metadata, the amount of the (required) original data, which is required to transfer to reducers only at required iterations, and different locations of data and computations.

\medskip \noindent \textbf{The communication cost for metadata and data.} In the context of Meta-MapReduce, the communication cost is the sum of the following:
\begin{description}[nolistsep,noitemsep]
  \item[Metadata cost] The amount of metadata that is required to move from the location of users to the location of mappers (if the locations of data and mappers are different) and from the map phase to the reduce phase in each iteration of MapReduce job.
  \item[Data cost] The amount of required original data that is needed to move to reducers at required iterations of a MapReduce job.
\end{description}

In the next Section~\ref{subsec:Meta-MapReduce framework}, we explain the way Meta-MapReduce works, using an example of equijoin for a case of different locations of data and mappers. Following the example of equijoin, we also show how much communication cost is reduced due to the use of Mata-MapReduce.% and present a tradeoff.

%\medskip \noindent \textbf{Meta-MapReduce framework (when data is at the site of user).}
\subsection{Meta-MapReduce Working}
\label{subsec:Meta-MapReduce framework}
\begin{figure*}[!t]
\centering
\includegraphics[scale=0.5]{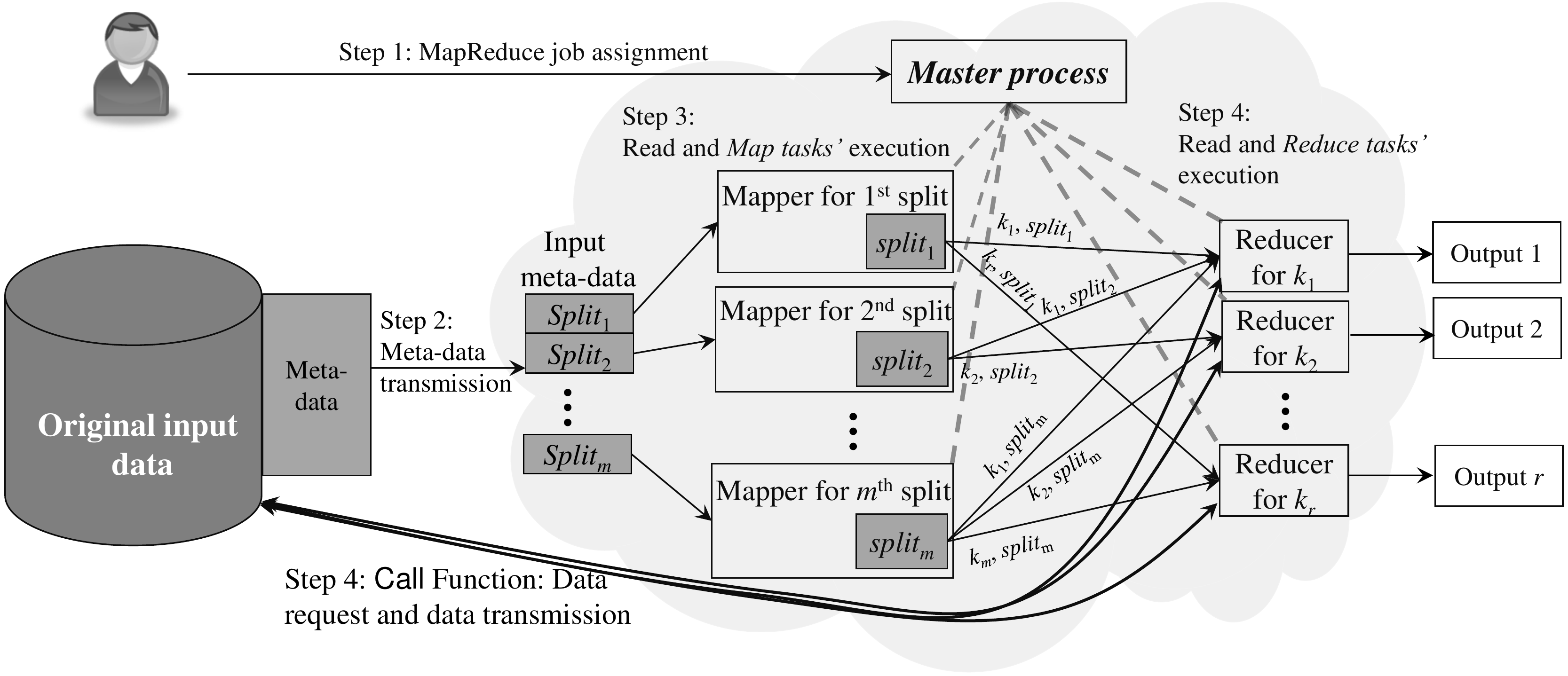}
\caption{Meta-MapReduce algorithmic approach.}
\label{fig:meta-mapreduce_cloud}
\end{figure*}

In the standard MapReduce, users send their data to the site of the mappers before the computation begins. However, in Meta-MapReduce, users send metadata to the site of mappers, instead of original data, see Figure~\ref{fig:meta-mapreduce_cloud}. Now, mappers and reducers work on metadata, and at required iterations of a MapReduce job, reducers \textit{call} required original data from the site of users (according to assigned $\langle \mathit{key, value}\rangle$ pairs) and provide the desired result. We present a detailed execution to demonstrate Meta-MapReduce (see Figure~\ref{fig:meta-mapreduce_cloud}), using the equijoin task, as follows:

\parskip 0pt
\setlength{\parindent}{15pt}

\begin{description}
  \item [\textsc{Step} 1] Users create a \emph{master process} that creates \emph{map tasks} and \emph{reduce tasks} at different compute nodes. A compute node that processes the \emph{map task} is called a \emph{map worker}, and a compute node that processes the \emph{reducer task} is called a \emph{reduce worker}.

  \item [\textsc{Step} 2] Users send metadata, which varies according to an assigned MapReduce job, to the site of mappers. Also, the user creates an index, which varies according to the assigned job, on the entire database.

  For example, in the case of equijoin (see Figure~\ref{fig:join_mr}), a user sends metadata for each of the tuples of the relations $X(A,B)$ and $Y(B,C)$ to the site of mappers. In this example, metadata for a tuple $i$ ($\langle a_i,b_i\rangle$, where $a_i$ and $b_i$ are values of the attributes $A$ and $B$, respectively) of the relation $X$ includes the size of all non-joining values (\textit{i}.\textit{e}., $|a_i|$\footnote{The notation $|a_i|$ refers to the size of an input $a_i$.}) and the value of $b_i$. Similarly, metadata for a tuple $i$ ($\langle b_i,c_i\rangle$, where $b_i$ and $c_i$ are values of the attributes $B$ and $C$, respectively) of the relation $Y$ includes the size of all non-joining values (\textit{i}.\textit{e}., $|c_i|$) with $b_i$ (remember that the size of $b_i$ is much smaller than the size of $a_i$ or $c_i$). In addition, the user creates an index on the attribute $B$ of both the relations $X$ and $Y$.

  \item [\textsc{Step} 3] In the map phase, a mapper processes an assigned input and provides some number of $\langle \mathit{key, value}\rangle$ pairs, which are known as \textit{intermediate outputs}, a $\mathit{value}$ is the \textit{size} of the corresponding input data (which is included in metadata). The \emph{master process} is then informed of the location of intermediate outputs.

      For example, in case of equijoin, a mapper takes a single tuple $i$ (\textit{e}.\textit{g}., $\langle |a_i|,b_i\rangle$) and generates some $\langle b_i, \mathit{value}\rangle$ pairs, where $b_i$ is a key and a $\mathit{value}$ is the \textit{size} of tuple $i$ (\textit{i}.\textit{e}., $|a_i|$). Note that in the original equijoin example, a $\mathit{value}$ is the whole data associated with the tuple $i$ (\textit{i}.\textit{e}., $a_i$).

  \item [\textsc{Step} 4] The \emph{master process} assigns \emph{reduce tasks} (by following a mapping schema as suggested in~\cite{DBLP:journals/corr/AfratiDK0U15a}) and provides information of intermediate outputs, which serve as inputs to \emph{reduce tasks}. A reducer is then assigned all the $\langle \mathit{key, value}\rangle$ pairs having an identical key by following a mapping schema for an assigned job. Now, reducers perform the computation and \texttt{call}\footnote{The \texttt{call} operation will be explained in Section~\ref{subsec:the call function}.} only required data if there is only one iteration of a MapReduce job. On the other hand, if a MapReduce job involves more than one iteration, then reducers \texttt{call} original required data at required iterations of the job (we will discuss multi-round MapReduce jobs using Meta-MapReduce in Section~\ref{subsec:Multi-round Iterations}).

     For example, in case of equijoin, a reducer receives all the $\langle \mathit{b_i, value}\rangle$ pairs from both the relations $X$ and $Y$, where a $\mathit{value}$ is the size of tuple associated with key $b_i$. Inputs (\textit{i}.\textit{e}., intermediate outputs of the map phase) are assigned to reducers by following a mapping schema for equijoin such that a reducer does not assign more original inputs than its capacity, and after that reducers invoke the \texttt{call} operation. Note that a reducer that receives at least one tuple with key $b_i$ from both the relations $X$ and $Y$ produces outputs and requires original input data from the user's site. However, if the reducer receives tuples with key $b_i$ from a single relation only, the reducer does not request for the original input tuple, since these tuples do not participate in the final output.
\end{description}
Following Meta-MapReduce, we now compute the communication cost involved in equijoin example (see Figure~\ref{fig:join_mr}). Recall that without using Meta-MapReduce, a solution to equijoin problem (in Figure~\ref{fig:join_mr}) requires 12 units communication cost. However, using Meta-MapReduce for performing equijoin, there is no need to send the tuple $\langle a_3,b_2\rangle$ of the relation $X$ and the tuples $\langle b_3,c_3\rangle$ of the relation $Y$ to the location of computation. Moreover, we send metadata of all the tuples to the site of mappers, and intermediate outputs containing metadata are transferred to the reduce phase, where reducers \texttt{call} only desired tuples having $b_1$ value from the user's site. Consequently, a solution to the problem of equijoin has only 4 units cost plus a constant cost for moving metadata using Meta-MapReduce, instead of 12 units communication cost.

%\noindent\textbf{Explanation of tradeoff.} We explain the tradeoff using equijoin example. In the case of original input data, users have to move data from their location to the location of mappers, and then, intermediate outputs are transferred from the map phase to reduce phase; such data transfer shows a high amount of data movement twice; however, reducers perform only one-time read operation for providing outputs. On the other hand, using Meta-MapReduce, we can obtain outputs by transferring less amount of data (\textit{i}.\textit{e}., metadata of each tuple), however, reducers perform two-times read operation to read metadata, and then, assigned original inputs.

\begin{theorem}[The communication cost]
\label{th:communication cost for two relations}
Using Meta-MapReduce, the communication cost for the problem of join of two relations is at most $2nc+h(c+w)$ bits, where $n$ is the number of tuples in each relation, $c$ is the maximum size of a value of the joining attribute, $h$ is the number of tuples that actually join, and $w$ is the maximum required memory for a tuple.
\end{theorem}

\begin{proof}
Since the maximum size of a value of the joining attribute, which works as a metadata in the problem of join, is $c$ and there are $n$ tuples in each relation, users have to send at most $2nc$ bits to the site of mappers-reducers. Further, tuples that join at the reduce phase have to be transferred from the map phase to the reduce phase and then from the user's site to the reduce phase. Since there are at most $h$ tuples join and the maximum size of a tuple is $w$, we need to transfer at most $hc$ and at most $hw$ bits from the map phase to the reduce phase and from the user's site to the reduce phase, respectively. Hence, the communication cost is at most $2nc+h(c+w)$ bits.
\end{proof}

\textbf{Further significant improvement.} We note that it is possible to further decrease the communication cost by using two iterations of a MapReduce job, in which the first iteration is performed on metadata and the second iteration is performed on the required original data. Specifically, in the first iteration, a user sends metadata to some reducers such that reducers are not assigned \textit{more metadata} than capacity, and all these reducers will compute the required original data and the optimal number of reducers for a task. Afterward, a new MapReduce iteration is executed on the required original data such that a reducer is not assigned \textit{more original inputs} than its capacity. In this manner, we save replication of metadata and some reducers that do not produce outputs.

\begin{table*}[t]
\begin{center}
\bgroup
\def\arraystretch{1.5}
    \centering

    \begin{tabular}{|l|l|l|l|l|}
    \hline
   \multirow{2}{6cm}{Problems} & \multirow{2}{*}{Section} & \multirow{2}{*}{Theorem} & \multicolumn{2}{ c| }{Communication cost}  \\

   \hhline{~~~|-|-|}  & {~} & {~}  & \textcolor[rgb]{0.00,0.50,0.00}{using Meta-MapReduce} & \textcolor[rgb]{1.00,0.00,0.00}{using MapReduce}  \\ \hline\hline

    Join of two relations & ~\ref{subsec:Meta-MapReduce framework} & ~\ref{th:communication cost for two relations} & $\textcolor[rgb]{0.00,0.50,0.00}{2nc+h(c+w)}$ & $\textcolor[rgb]{1.00,0.00,0.00}{4nw}$ \\ \hline

    Skewed Values of the Joining Attribute & ~\ref{subsec:Meta-MapReduce for Skewed Values of the Joining Attribute} & ~\ref{th:communication cost for skew join} & $\textcolor[rgb]{0.00,0.50,0.00}{2nc + rh(c + w)}$ & $\textcolor[rgb]{1.00,0.00,0.00}{2nw(1+r)}$\\ \hline

    Join of two relations by hashing the joining attribute  & ~\ref{subsec:Large Size of Joining Values} & ~\ref{th:communication cost heavy size of joining value} & $\textcolor[rgb]{0.00,0.50,0.00}{6n\cdot log\ m+h(c+w)}$ & $\textcolor[rgb]{1.00,0.00,0.00}{4nw}$\\ \hline

    Join of $k$ relations by hashing the joining attributes & ~\ref{subsec:Multi-round Iterations} & ~\ref{th:communication cost multi round} & $\textcolor[rgb]{0.00,0.50,0.00}{3knp\cdot log\ m +h(c+w)}$ & $\textcolor[rgb]{1.00,0.00,0.00}{2knw}$\\ \hline\hline
    \multicolumn{5}{|p{\linewidth}|}{$n$: the number of tuples in each relation, $c$: the maximum size of a value of the joining attribute, $r$: the replication rate, $h$: the number of tuples that actually join, $w$ is the maximum required memory for a tuple, $p$: the maximum number of dominating attributes in a relation, and $m$: the maximal number of tuples in all given relations.}\\\hline

    \end{tabular}
    \egroup
   \B
\caption{The communication cost for joining of relations using Meta-MapReduce.}
\label{table:The communication cost for joining of relations using Meta-MapReduce}
\BBB
\end{center}
\end{table*}

\subsection{The \texttt{Call} Function}
\label{subsec:the call function}
In this section, we will describe the \texttt{call} function that is invoked by reducers to have the original required inputs from the user's site to produce outputs.

\parskip 0pt
\setlength{\parindent}{15pt}

All the reducers that produce outputs require the original inputs from the site of users. Reducers can know whether they produce outputs or not, after receiving intermediate outputs from the map phase, and then, inform the corresponding mappers from where they have fetched these intermediate outputs (for simplicity, we can say all reducers that will produce outputs send 1 to all the corresponding mappers to request the original inputs, otherwise send 0). Mappers collect requests for the original inputs from all the reducers and fetch the original inputs, if required, from the user's site by accessing the index file. Remember that in Meta-MapReduce, the user creates an index on the entire database according to an assigned job, refer to \textsc{Step} 2 in Section~\ref{subsec:Meta-MapReduce framework}. This index helps to access required data that reducers want without doing a scan operation. Note that the \texttt{call} function can be easily implemented on recent implementations of MapReduce, \textit{e}.\textit{g}., Pregel and Spark.

For example, we can consider our running example of equijoin. In case of equijoin, a reducer that receives at least one tuple with key $b_i$ from both the relations $X(A,B)$ and $Y(B,C)$ requires the original input from the user's site, and hence, the reducer sends 1 to the corresponding mappers. However, if the reducer receives tuples with key $b_i$ from a single relation only, the reducer sends 0. Consider that the reducer receives $\langle b_i, |a_i|\rangle$ of the relation $X$ and $\langle b_i, |c_i|\rangle$ of the relation $Y$. The reducer sends 1 to corresponding mappers that produced $\langle b_i, |a_i|\rangle$ and $\langle b_i, |c_i|\rangle$ pairs. On receiving requests for the original inputs from the reducer, the mappers access the index file to fetch $a_i$, $b_i$, and $c_i$, and then, the mapper provides $a_i$, $b_i$, and $c_i$ to the reducer.

\subsection{Meta-MapReduce for Skewed Values of the Joining Attribute}
\label{subsec:Meta-MapReduce for Skewed Values of the Joining Attribute}
%\medskip\medskip\noindent\textbf{Meta-MapReduce for skewed values of joining attribute.}
Consider two relations $X(A,B)$ and $Y(B,C)$, where the joining attribute is $B$ and the size of all the $B$ values is very small as compared to the size of values of the attributes $A$ and $C$. One or both of the relations $X$ and $Y$ may have a large number of tuples with an identical $B$-value. A value of the joining attribute $B$ that occurs many times is known as a {\em heavy hitter}. In skew join of $X(A,B)$ and $Y(B,C)$, all the tuples of both the relations with an identical heavy hitter should appear together to provide the output tuples.

\parskip 0pt
\setlength{\parindent}{15pt}

In Figure~\ref{fig:twowayjoin}, $b_1$ is a heavy hitter; hence, it is required that all the tuples of $X(A,B)$ and $Y(B,C)$ with the heavy hitter, $b_1$, should appear together to provide the output tuples, $\langle a, b_1, c\rangle$ ($a\in A, b_1\in B, c\in C$), which depend on exactly two inputs. However, due to a single reducer --- for joining all tuples with a heavy hitter --- there is no parallelism at the reduce phase, and a single reducer takes a long time to produce all the output tuples of the heavy hitter.

\begin{figure}[!h]
\begin{center}
\includegraphics[width=65mm, height=35mm]{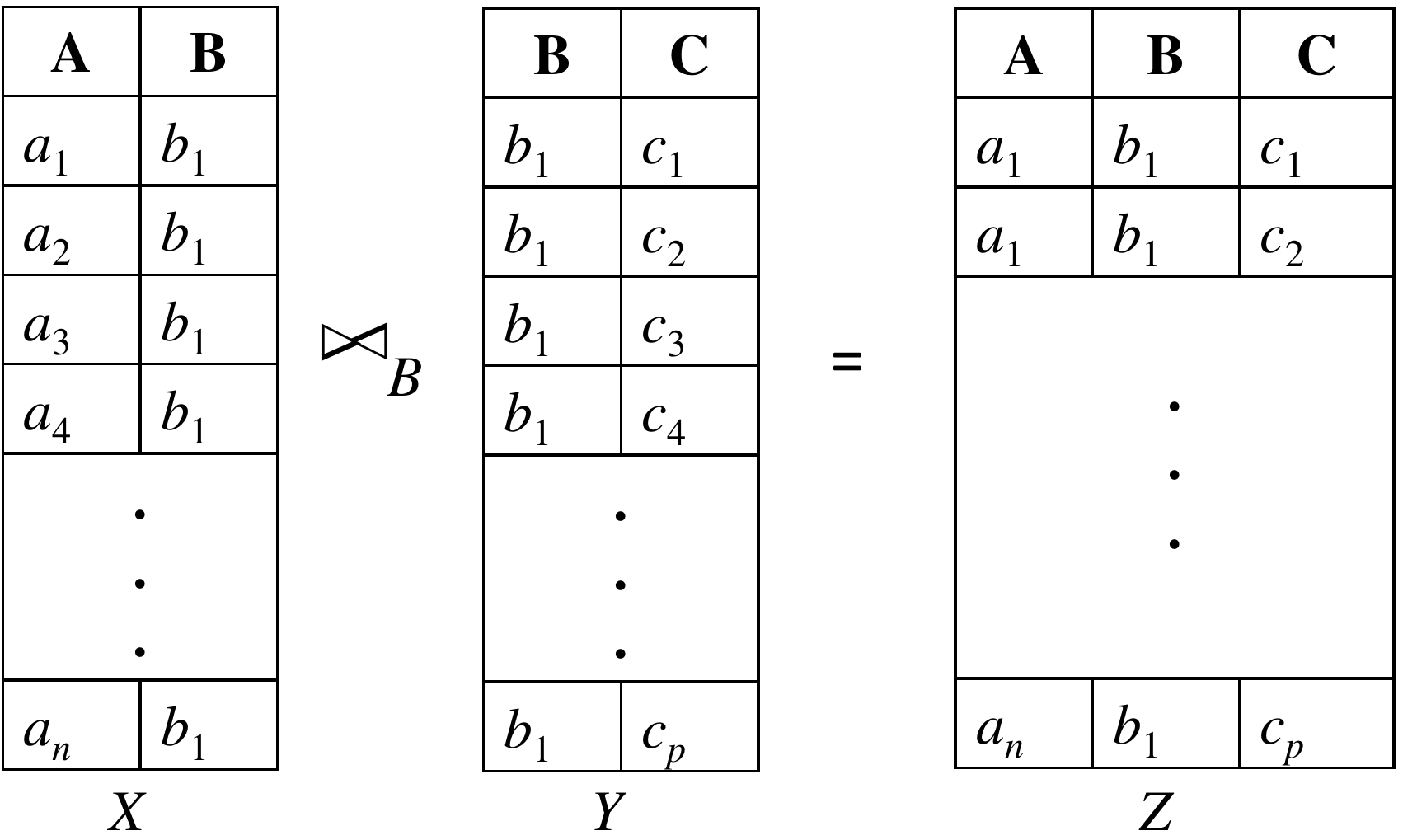}
\end{center}
\caption{Skew join example for a heavy hitter, $b_1$.}
\label{fig:twowayjoin}
\end{figure}
We can restrict reducers in a way that they can hold many tuples, but not all the tuples with the heavy-hitter-value. In this case, we can reduce the time and use more reducers, which results in a higher level of parallelism at the reduce phase. But, there is a higher communication cost, since each tuple with the heavy hitter must be sent to more than one reducer. (Details about join algorithms may be found in~\cite{DBLP:journals/tkde/AfratiU11}.)

We can solve the problem of skew join using Meta-MapReduce, using four steps suggested in Section~\ref{subsec:Meta-MapReduce framework}.

\begin{theorem}[The communication cost]
\label{th:communication cost for skew join}
Using Meta-MapReduce, the communication cost for the problem of skew join of two relations is at most $2nc+rh(c+w)$ bits, where $n$ is the number of tuples in each relation, $c$ is the maximum size of a value of the joining attribute, $r$ is the replication rate,\footnote{The replication rate~\cite{DBLP:journals/corr/abs-1206-4377} of a mapping schema is the average number of key-value pairs for each input.} $h$ is the number of distinct tuples that actually join, and $w$ is the maximum required memory for a tuple.
\end{theorem}

\begin{proof}
From the user's site to the site of mappers-reducers, at most $2nc$ bits are required to move (according to Theorem~\ref{th:communication cost for two relations}). Since at most $h$ distinct tuples join and these tuples are replicated to $r$ reducers, at most $rhc$ bits are required to transfer from the map phase to the reduce phase. Further, $h$ tuples of size at most $w$ to be transferred from the map phase to the reduce phase, and hence, at most $rhw$ bits are assigned to reducers. Thus, the communication cost is at most $2nc+rh(c+w)$ bits.
\end{proof}

\subsection{Meta-MapReduce for an Identical Location of Data and Mappers}
%\label{subsec:Meta-MapReduce framework when data is at the site of mappers}
%\medskip\medskip \noindent\textbf{Meta-MapReduce for an identical location of data and mappers.}
We explained the way Meta-MapReduce acts in the case of different locations of data and computation, and show how it provides desired outputs by considering only metadata of inputs. Nevertheless, Meta-MapReduce also decreases the amount of data to be transferred when the locations of data and computations are identical. In this case, mappers process only metadata of assigned inputs instead of the original inputs as in MapReduce, and provide $\langle \mathit{key, value}\rangle$ pairs, where a $\mathit{value}$ is the size of an assigned input, not the original input itself. A reducer processes all the $\langle \mathit{key, value}\rangle$ pairs having an identical key and calls the original input data, if required. Consequently, there is no need to send all those inputs that do not participate in the final result from the map phase to the reduce phase.

\parskip 0pt
\setlength{\parindent}{15pt}

For example, in case of equijoin, if the location of mappers and relations $X$ and $Y$ are identical, then a mapper processes a tuple of either relation and provides $\langle \mathit{b_i, |a_i|}\rangle$ or $\langle \mathit{b_i, |c_i|}\rangle$ as outputs. A reducer is assigned all the inputs having an identical key, and the reducer calls the original inputs if it has received inputs from both the relations.

\section{Extensions of Meta-MapReduce}
\label{sec:Extensions to Meta-MapReduce}
We have presented Meta-MapReduce framework for different and identical locations of data and mappers-reducers. However, some extensions are required to use Meta-MapReduce for geographically distributed data processing, for handling large size of values of joining attributes, and for handling multi-round iterations. In this section, we will provide three extensions of Meta-MapReduce.

%\medskip\medskip \noindent\textbf{Incorporating Meta-MapReduce in G-Hadoop and Hierarchical MapReduce.}
\subsection{Incorporating Meta-MapReduce in G-Hadoop and Hierarchical MapReduce}
\label{subsec:Incorporating Meta-MapReduce in G-Hadoop and Hierarchical MapReduce}
G-Hadoop and Hierarchical MapReduce are two implementations for geographically distributed data processing using MapReduce. Both the implementations assume that a cluster processes data using MapReduce and provides its outputs to one of the clusters that provides final outputs (by executing a MapReduce job on the received outputs of all the clusters). However, the transmission of outputs of all the clusters to a single cluster for producing the final output is not efficient, if all the outputs of a cluster do not participate in the final output.

\parskip 0pt
\setlength{\parindent}{15pt}

We can apply Meta-MapReduce idea to systems such as G-Hadoop and Hierarchical MapReduce. Note that we \textit{do not} change basic functionality of both the implementations. We take our running example of equijoin (see Figure~\ref{fig:Three clusters each with two relations}, where we have three clusters, possibly on three continents, the first cluster has two relations $U(A,B)$ and $V(B,C)$, the second cluster has two relations $W(D,B)$ and $X(B,E)$, and the third cluster has two relations $Y(F,B)$ and $Z(B,G)$) and assume that data exist at the site of mappers in each cluster. In the final output, reducers perform the join operation over all the six relations, which share an identical $B$-value.

\begin{figure}[!h]
\begin{center}
\includegraphics[scale=0.37]{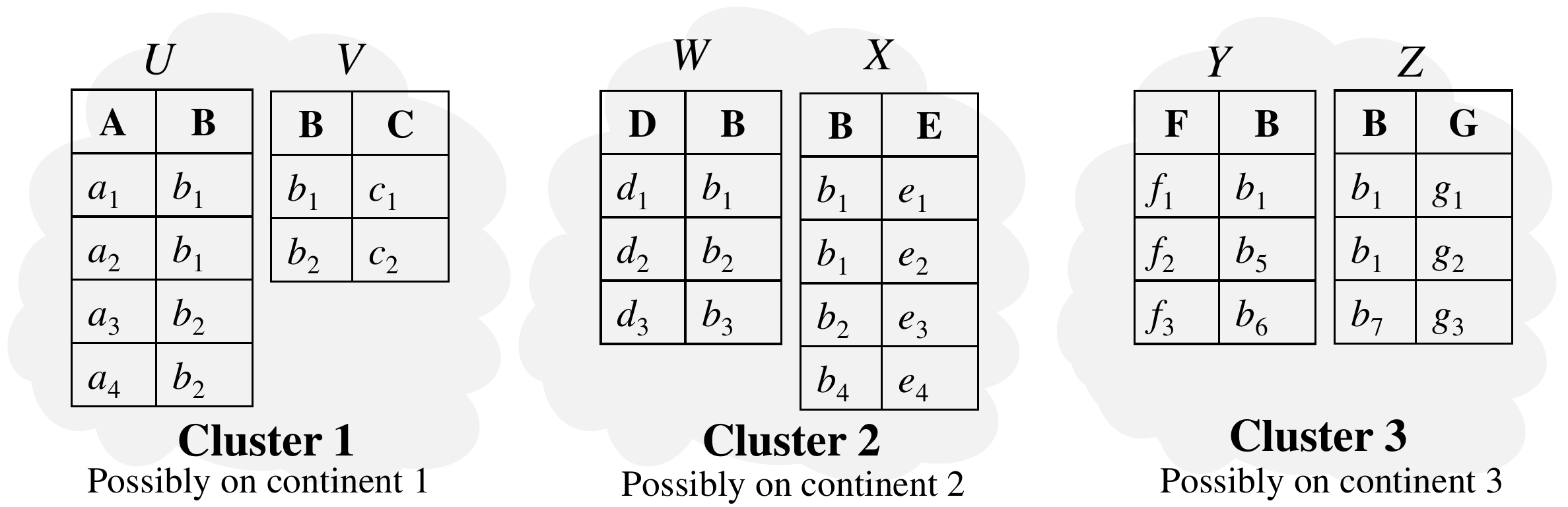}
\end{center}
\caption{Three clusters, each with two relations.}
\label{fig:Three clusters each with two relations}
\end{figure}

The following three steps are required for obtaining final outputs using an execution of Meta-MapReduce over G-Hadoop and Hierarchical MapReduce.
\begin{description}[nolistsep,noitemsep]
  \item[\textsc{Step} 1] Mappers at each cluster process input data according to an assigned job and provide $\langle \mathit{key, value} \rangle$ pairs, where a $\mathit{value}$ is the size of an assigned input.

For example, in Figure~\ref{fig:Three clusters each with two relations}, a mapper at Cluster 1 provides outputs of the form of $\langle \mathit{b_i, |a_i|} \rangle$ or $\langle \mathit{b_i, |c_i|} \rangle$.

%      For example, in Figure~\ref{fig:Three clusters each with two relations}, a mapper at Cluster 1 processes a tuple $i$ ($\langle a_i,b_i\rangle$) of the relation $U$ or a tuple $i$ ($\langle b_i,c_i \rangle$) of the relation $V$ and provides outputs of the form of $\langle \mathit{b_i, |a_i|} \rangle$ or $\langle \mathit{b_i, |c_i|} \rangle$.

  \item[\textsc{Step} 2] Reducers at each cluster provide partial outputs by following an assigned mapping schema, and partial outputs, which contain only metadata, are transferred to one of the clusters, which will provide final outputs.

      For example, in case of equijoin, reducers at each cluster provide partial output tuples as $\langle \mathit{|a_i|,b_i,|c_i|} \rangle$ at Cluster 1, $\langle \mathit{|d_i|,b_i,|e_i|} \rangle$ at Cluster 2, and $\langle \mathit{|f_i|,b_i,|g_i|} \rangle$ at Cluster 3 (by following a mapping schema for equijoin). Partial outputs of Cluster 1 and Cluster 3 have to be transferred to one of the clusters, say Cluster 2, for obtaining the final output.

  \item[\textsc{Step} 3] A designated cluster for providing the final output processes all the outputs of the clusters by implementing the assigned job using Meta-MapReduce. Reducers that provide the final output \texttt{call} the original input data from all the clusters.

      For example, in equijoin, after receiving outputs of Cluster 1 and Cluster 3, Cluster 2 implements two iterations for joining tuples. In the first iteration, outputs of Clusters 1 and 2 are joined (by following a mapping schema for equijoin), and in the second iteration, outputs of Clusters 3 and the output of the previous iteration are joined at reducers. A reducer in the second iteration provides the final output as $\langle \mathit{|a_i|,b_1,|c_i|,|d_i|,|e_i|,|f_i|,|g_i|} \rangle$ and \texttt{call}s all the original values of $|a_i|$, $|c_i|$, $|d_i|$, $|e_i|$, $|f_i|$, and $|g_i|$ for providing the desired output, as suggested in Section~\ref{subsec:the call function}.
\end{description}

\medskip\noindent\textbf{Communication cost analysis.} In Figure~\ref{fig:Three clusters each with two relations}, we are performing equijoin in three clusters, and assuming that data is available at the site of mappers in each cluster. In addition, we consider that each value takes two units size; hence, any tuple, for example, $\langle a_i, b_i\rangle$, has size 4 units.

First, each of the clusters performs an equijoin within the cluster using Meta-MapReduce. Note that using Meta-MapReduce, there is no need to send any tuple from the map phase to the reduce phase within the cluster, while G-Hadoop and Hierarchical MapReduce do data transfer from the map phase to the reduce phase, and hence, results in 76 units of communication cost. Moreover, in G-Hadoop and Hierarchical MapReduce, the transmission of two tuples ($\langle a_3, b_2\rangle$, $\langle a_4,b_2\rangle$) of $U$, one tuple ($\langle b_2,c_2\rangle$) of $V$, two tuples ($\langle d_2,b_2\rangle$, $\langle d_3,b_3\rangle$) of $W$, two tuples ($\langle b_2,e_3 \rangle$, $\langle b_4,e_4\rangle$) of $X$, two tuples ($\langle f_2,b_5\rangle$, $\langle f_3,b_6\rangle$) of $Y$, and one tuple ($\langle b_7,g_3\rangle$) of $Z$ from the map phase to the reduce phase is useless, since they do not participate in the final output.

After computing outputs within the cluster, metadata of outputs (\textit{i}.\textit{e}., size of tuples associated with key $b_1$ and key $b_2$) is transmitted to Cluster 2. Here, it is important to note that tuples with value $b_1$ provide final outputs. Using Meta-MapReduce, we will not send the complete tuples with value $b_2$, hence, we also decrease the communication cost; while G-Hadoop and Hierarchical MapReduce send all the outputs of the first and third clusters to the second cluster. After receiving outputs from the first and the third clusters, the second cluster performs two iterations as mentioned previously, and in the second iteration, a reducer for key $b_1$ provides the final output. Following that the communication cost is only 36 units.

On the other hand, transmission of outputs with data from the first cluster and the third cluster to the second cluster and performing two iterations result in 132 units communication cost. Therefore, G-Hadoop and Hierarchical MapReduce require 208 units communication cost while Meta-MapReduce provides the final results using 36 units communication cost.

%\medskip\medskip\noindent\textbf{Heavy size of joining values.}
\subsection{Large Size of Joining Values}
\label{subsec:Large Size of Joining Values}
We have considered that sizes of joining values are very small as compared to sizes of all the other non-joining values. For example, in Figure~\ref{fig:join_mr}, sizes of all the values of the attribute $B$ are very small as compared to all the values of the attributes $A$ and $C$. However, considering very small size of values of the joining attribute is not realistic. All the values of the joining attribute may also require a considerable amount of memory, which may be equal or greater than the sizes of non-joining values. In this case, it is not useful to send all the values of the joining attribute with metadata of non-joining attributes. Thus, we enhance Meta-MapReduce for handling a case of large size of joining values.

\parskip 0pt
\setlength{\parindent}{15pt}

We consider our running example of join of two relations $X(A,B)$ and $Y(B,C)$, where the size of each of $b_i$ is large enough such that the value of $b_i$ cannot be used as metadata. We use hash function to gain a short identifier (that is unique with high probability) for each $b_i$. We denote $H(b_i)$ to be the hash value of the original value of $b_i$. Here, Meta-MapReduce works as follows:

\begin{description} [nolistsep,noitemsep]
  \item[\textsc{Step} 1] For all the values of the joining attribute ($B$), use a hash function such that an identical $b_i$ in both of the relations has a unique hash value with a high probability, and $b_i$ and $b_j$, $i\neq j$, receive two different hash values with a high probability.

  \item[\textsc{Step} 2] For all the other non-joining attributes' values (values corresponding to the attributes $A$ and $C$), find metadata that includes size of each of the values.

  \item[\textsc{Step} 3] Perform the task using Meta-MapReduce, as follows: (\textit{i}) Users send hash values of joining attributes and metadata of the non-joining attributes. For example, a user sends hash value of $b_i$ ($H(b_i)$) and the corresponding metadata (\textit{i}.\textit{e}., size) of values $a_i$ or $c_i$ to the site of mappers. (\textit{ii}) A mapper processes an assigned tuples and provides intermediate outputs, where a $\mathit{key}$ is $H(b_i)$ and a $\mathit{value}$ is $|a_i|$ or $|c_i|$. (\textit{iii}) Reducers \texttt{call} all the $\mathit{value}$s corresponding to a key (hash value), and if a reducer receives metadata of $a_i$ and $c_i$, then the reducer calls the original input data and provides the final output.

      Note that there may be a possibility that two different values of the joining attribute have an identical hash value; hence, these two values are assigned to a reducer. However, the reducer will know these two different values, when it will \texttt{call} the corresponding data. The reducer notifies the master process, and a new hash function is used.
\end{description}

\begin{theorem}[The communication cost]
\label{th:communication cost heavy size of joining value}
Using Meta-MapReduce for the problem of join where values of joining attributes are large, the communication cost for the problem of join of two relations is at most $6n\cdot log\ m+h(c+w)$ bits, where $n$ is the number of tuples in each relation, $m$ is the maximal number of tuples in two relations, $h$ is the number of tuples that actually join, and $w$ is the maximum required memory for a tuple.
\end{theorem}
\begin{proof}
The maximal number of tuples having different values of a joining attribute in all relations is $m$, which is upper bounded by $2n$; hence, a mapping of hash function of $m$ values into $m^3$ values will result in a unique hash value for every of the $m$ keys with a high probability. Thus, we use at most $3\cdot log\ m$ bits for metadata of a single value, and hence, at most $6n\cdot log\ m$ bits are required to move metadata from the user's site to the site of mappers-reducers. Since there are at most $h$ tuples join and the maximum size of a tuple is $w$, we need to transfer at most $hc$ and at most $hw$ bits from the map phase to the reduce phase and from the user's site to the reduce phase, respectively. Hence, the communication cost is at most $6n\cdot log\ m+h(c+w)$ bits.
\end{proof}

%\medskip\medskip\noindent\textbf{Multiround iterations.}
\subsection{Multi-round Iterations}
\label{subsec:Multi-round Iterations}
We show how Meta-MapReduce can be incorporated in a multi-round MapReduce job, where values of joining attributes are also large as the value of non-joining attributes. In order to explain, the working of Meta-MapReduce in a multi-iterative MapReduce job, we consider an example of join of four relations $U(A,B,C,D)$, $V(A,B,D,E)$, $W(D,E,F)$, and $X(F,G,H)$, and perform the join operation using a cascade of two-way joins..

\parskip 0pt
\setlength{\parindent}{15pt}

\begin{description}[noitemsep,nolistsep]
  \item[\textsc{Step} 1] Find \emph{dominating attributes} in all the relations. An attribute that occurs in more than one relation is called a dominating attribute~\cite{DBLP:journals/tkde/AfratiU11}.

      For example, in our running example, attributes $A$, $B$, $D$, $E$, and $F$ are dominating attributes.

  \item[\textsc{Step} 2] Implement a hash function over all the values of dominating attributes so that all the identical values of dominating attributes receive an identical hash value with a high probability, and all the different values of dominating attributes receive different hash values with a high probability.

      For example, identical values of $a_i$, $b_i$, $d_i$, $e_i$, and $f_i$ receive an identical hash value, and any two values $a_i$ and $a_j$, such that $i\neq j$, probably receive different hash values (a similar case exists for different values of attributes $B$, $D$, $E$, $F$).

  \item[\textsc{Step} 3] For all the other non-dominating joining attributes' (an attribute that occurs in only one of the relations) values, we find metadata that includes size of each of the values.

  \item[\textsc{Step} 4] Now perform 2-way cascade join using Meta-MapReduce and follow a mapping schema according to a problem for assigning inputs (\textit{i}.\textit{e}., outputs of the map phase) to reducers.

  For example, in equijoin example, we may join relations as follows: first, join relations $U$ and $V$, and then join the relation $W$ to the outputs of the join of relations $U$ and $V$. Finally, we join the relation $X$ to outputs of the join of relations $U$, $V$, and $W$. Thus, we join the four relations using three iterations of Meta-MapReduce, and in the final iteration, reducers \texttt{call} the original required data.
\end{description}

\noindent \textit{Example}: Following our running example, in the first iteration, a mapper produces $\langle H(a_i), [H(b_i), |c_i|, H(d_i)]\rangle$ after processing a tuple of the relation $U$ and $\langle H(a_i), [H(b_i), H(d_i),H(e_i)]\rangle$ after processing a tuple of the relation $V$ (where $H(a_i)$ is a key). A reducer corresponding to $H(a_i)$ provides $\langle H(a_i), H(b_j), |c_k|, H(d_l), H(e_z)\rangle$ as outputs.

In the second iteration, a mapper produces $\langle H(d_i), [H(a_i), H(b_j), |c_k|, H(e_z)]\rangle$ and $\langle H(d_i), [H(e_i), H(f_i)]\rangle$ after processing outputs of the first iteration and the relation $W$, respectively. Reducers in the second iteration provide output tuples by joining tuples that have an identical $H(d_i)$. In the third iterations, a mapper produces $\langle H(f_i), [H(a_i), H(b_i), |c_i|, H(d_i), H(e_i)]\rangle$ or $\langle H(f_i), [|g_i|,|h_i|]\rangle$, and reducers perform the final join operations. A reducer, for key $H(f_i)$, receives $|g_i|$ and $|h_i|$ from the relation $X$ and output tuples of the second iteration, provides the final output by calling original input data from the location of user.

\begin{theorem}[The communication cost]
\label{th:communication cost multi round}
Using Meta-MapReduce for the problem of join where values of joining attributes are large, the communication cost for the problem of join of $k$ relations, each of the relations with $n$ tuples, is at most $3knp\cdot log\ m +h(c+w)$ bits, where $n$ is the number of tuples in each relation, $m$ is the maximal number of tuples in $k$ relations, $p$ is the maximum number of dominating attributes in a relation, $h$ is the number of tuples that actually join, and $w$ is the maximum required memory for a tuple.
\end{theorem}
\begin{proof}
According to Theorem~\ref{th:communication cost heavy size of joining value}, at most $3\cdot log\ m$ bits for metadata are required for a single value; hence, at most $3knp\cdot log\ m$ bits are required to move metadata from the user's site to the site of mappers-reducers. Since at most $h$ tuples join and the maximum size of a tuple is $w$, at most $hc$ and at most $hw$ bits from the map phase to the reduce phase and from the user's site to the reduce phase, respectively, are transferred. Hence, the communication cost is at most $3knp\cdot log\ m +h(c+w)$ bits.
\end{proof}

%\subsection{Incorporating the Reducer Capacity}
%Capacity of a reducer is the maximum size of inputs that can be assigned to that reducer. We follow a similar technique for regarding the reducer capacity as suggest in~\cite{TR}. Here, we explain in short how Meta-MapReduce does not assign more inputs to a reducer than its capacity. Recall that in Meta-MapReduce, mappers process metadata, which contain sizes of inputs, and generate some key-value pairs; in order to follow the reducer capacity, mappers generate key-value pairs according to assigned problems and following algorithms given in~\cite{TR}. Following these algorithms, no reducer is assigned more input sizes than its capacity. After that, reducers provide final outputs, as suggested in Section~\ref{subsec:Meta-MapReduce framework}. Consequently, Meta-MapReduce is able to regard the reducer capacity, and solutions to problem where an output depended on at least two inputs.

\section{Versatility of Meta-MapReduce}
\label{sec:Versatility of Meta-MapReduce}
Meta-MapReduce decreases the amount of data to be transferred to the remote site and intermediate data, if the final output does not depend on all the inputs. Especially, the problems, where the amount of intermediate data is much larger than the input data to the map phase and all the inputs do not participate in the final output, fit well in the context of Meta-MapReduce. In this section, we provide two common problems that can be solved using Meta-MapReduce.

\parskip 0pt
\setlength{\parindent}{15pt}

%\subsection{$k$-NN Problem using Meta-MapReduce}
%\label{subsec:k-NN problem using Meta-MapReduce}
\medskip\noindent\textbf{$k$-NN Problem using Meta-MapReduce.}
\emph{Problem statement}: $k$-nearest-neighbors ($k$-NN) problem~\cite{DBLP:conf/edbt/ZhangLJ12,DBLP:journals/pvldb/LuSCO12} tries to find $k$-nearest-neighbors of a given object. Two relations $R$ of $m$ tuples and $S$ of $n$ tuples are inputs to the $k$-NN problem, where $m<n$. For example, relations $R$ and $S$ may contain a list of cities with full description of the city, images of places to visit in the city. Following that a solution to the $k$-NN problem finds $k$ cities from the relation $S$ for each city of the relation $R$ in a manner that the distance between two cities (the first city belongs to $R$ and the second city belongs to $S$) is minimum, and hence, $km$ pairs are produced as outputs. A basic approach to find $k$-NN is given in~\cite{DBLP:conf/edbt/ZhangLJ12} that uses two iterations of MapReduce, where the first iteration provides local $k$-NN; and the second iteration provides the global $k$-NN for each tuple of $R$.

\emph{Communication cost analysis}: %Consider that the locations of relations and mappers-reducers are different.
The communication cost is the size of all the tuples of $R$ and $S$ that are required to move to the location of mappers-reducers, and then, tuples from the map phase to the reduce phase in the two iterations. If $k\leq m$, then it is communication in-efficient to move all the tuples of $S$ from the user's location to the location of mappers and from the map phase to the reduce phase. Hence, by sending only metadata of each tuple a lot of communication can be avoided.

%\smallskip\noindent\textit{$k$-Nearest-Neighbors and Similarity Search.} We wish to find pairs out of $km$ pairs (provided by a solution to the $k$-NN problem) that have almost similar distance between two tuples (or cities). In order to do, one more iteration is required after two iterations, which were used to compute $k$-NN. The communication cost in this case follows the communication cost of the $k$-NN problem and the size of (copies of) $km$ pairs to send from the map phase to the reduce phase in the third iteration. In this example, it is clear that it is communication inefficient to move all the $n$ tuples of the relation $S$ in the first two iterations and all the $mk$ tuples of the relations $R$ and $S$ in the third iteration, while the final output can be commutated using metadata of tuples of the relation $S$ and send data in the third iteration.

%\subsection{BFS Execution or Path Findings on a Social Networking Graph using Meta-MapReduce}
%\label{subsec:BFS Execution or Path Findings on a Social Networking Graph using Meta-MapReduce}
\begin{figure}[h]
\begin{center}
\includegraphics[width=65mm, height=20mm]{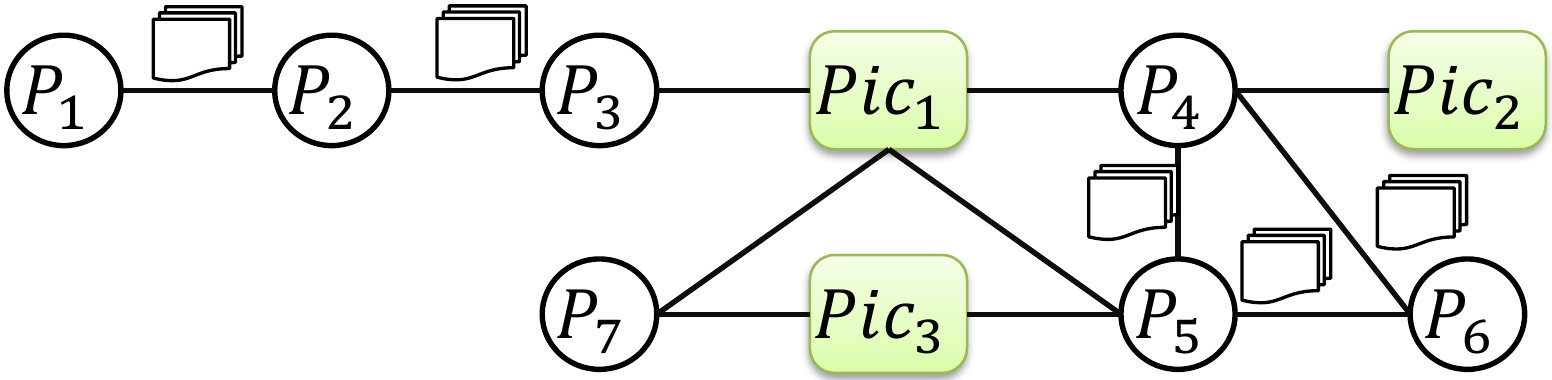}
\end{center}
\caption{A graph of a social network.}
\label{fig:social_graph}
\end{figure}

\medskip\noindent\textbf{Shortest Path Findings on a Social Networking Graph using Meta-MapReduce.} \emph{Problem statement}: Consider a graph of a social network, where a node represents either a person or a photo, and an edge exists between two persons if they are friends or between a person and a photo if the person is tagged in the photo; however, there is no edge between two photos; see Figure~\ref{fig:social_graph}. The implementation of a shortest path algorithm on the graph results in paths between two persons $i$ and $j$ with common information (which exist on the paths) between the two persons $i$ and $j$.

\emph{Communication cost analysis}: We want to find a shortest path between persons $P_1$ and $P_6$, refer to Figure~\ref{fig:social_graph}. A shortest path algorithm will provide a path $P_1$-$P_2$-$P_3$-$\mathit{Pic_1}$-$P_5$-$P_6$ or $P_1$-$P_2$-$P_3$-$\mathit{Pic_1}$-$P_4$-$P_6$ and show the common things between every two persons on the path and the photo, $\mathit{Pic_1}$, as a connection between $P_3$ and $P_5$. Note that in this case it is communication in-efficient to send all the photos and common information between every two friends of the graph to the site of mappers, because most of the nodes are removed in the final output. Hence, it is beneficial to send metadata of people and photos to the location of mappers and to process metadata to find a shortest path between two person at the reduce phase. In Figure~\ref{fig:social_graph}, reducers that provide the final output \texttt{call} information of person ($P_1, P_2, P_3, P_5, P_6$) and photo ($\mathit{Pic_1}$). Consequently, there is no need to send $\mathit{Pic_2}$ and $\mathit{Pic_3}$.

%\smallskip \textit{The opportunity of hashing}. In certain situations, there is a need to find identical big chunks of data, may it be photos, movies, pdf files, etc. Meta-MapReduce uses hash values during its operation and therefore may support an efficient treatment for such cases. For example, suppose we are given a set of photos, and we want to find whether two photos are identical or not. Here, it is required to compare every two photos at one reducer. However, it is not efficient to move all photos from the user's location to the location of computation if there are only a few photos that may be identical. Using Meta-MapReduce, we can use a hash function such that identical photos have a unique hash value with a high probability. Thus, we can send only hash values to the location of computation, and reducers that receives more than one hash value \texttt{call} the photos corresponding to the hash value and provide outputs.

%\subsection{Meta-MapReduce analysis and comparison with the classical MapReduce}
%\medskip\medskip \noindent\textbf{Meta-MapReduce and the classical MapReduce.}
%1. Fault tolerance. 2. Scalability. 3. Memory requirements. 4. Time complexity

\section{Conclusion}
Impacts of the locations of data and mappers-reducers, the number of iterations, and the reducer capacity (the maximum size of inputs that a reducer can hold) on the amount of data to be transferred to the location of (remote) computations are investigated. Based on the investigation, we found that it is not required to send the whole data to the location of computation if all the inputs do not participate in the final output. Thus, we proposed a new algorithmic technique for MapReduce algorithms, called Meta-MapReduce. Meta-MapReduce decreases a huge amount of data to be transferred across clouds by transferring metadata for a data field, rather than the field itself, metadata that is exponentially smaller, and processes metadata at the map phase and the reduce phase. We demonstrated the impact of Meta-MapReduce for solving problems of equijoin, $k$-nearest-neighbors finding, and shortest path finding. In the simplest case of equijoin of two relations, Meta-MapReduce requires at most $2nc+h(c+w)$ bits to be transferred to the location of computation, while the classical MapReduce requires at most $4nw$ bits to be transferred to the location of computation. Also, we suggest a way to incorporate Meta-MapReduce for processing geographically distributed data and for executing a multi-round MapReduce computation.

%\begin{spacing}{1}
%\bibliographystyle{abbrv}
%\bibliographystyle{IEEEtran}
%\bibliography{meta-MRRelatedWork-full}

\begin{thebibliography}{1}
\providecommand{\url}[1]{#1}
\csname url@samestyle\endcsname
\providecommand{\newblock}{\relax}
\providecommand{\bibinfo}[2]{#2}
\providecommand{\BIBentrySTDinterwordspacing}{\spaceskip=0pt\relax}
\providecommand{\BIBentryALTinterwordstretchfactor}{4}
\providecommand{\BIBentryALTinterwordspacing}{\spaceskip=\fontdimen2\font plus
\BIBentryALTinterwordstretchfactor\fontdimen3\font minus
  \fontdimen4\font\relax}
\providecommand{\BIBforeignlanguage}[2]{{%
\expandafter\ifx\csname l@#1\endcsname\relax
\typeout{** WARNING: IEEEtran.bst: No hyphenation pattern has been}%
\typeout{** loaded for the language `#1'. Using the pattern for}%
\typeout{** the default language instead.}%
\else
\language=\csname l@#1\endcsname
\fi
#2}}
\providecommand{\BIBdecl}{\relax}
\BIBdecl

\bibitem{DBLP:conf/osdi/DeanG04}
\BIBentryALTinterwordspacing
J.~Dean and S.~Ghemawat, ``{MapReduce}: Simplified data processing on large
  clusters,'' in \emph{6th Symposium on Operating System Design and
  Implementation {(OSDI} 2004), San Francisco, California, USA, December 6-8,
  2004}, 2004, pp. 137--150. [Online]. Available:
  \url{http://www.usenix.org/events/osdi04/tech/dean.html}
\BIBentrySTDinterwordspacing

\bibitem{DBLP:books/ullman2011}
J.~Leskovec, A.~Rajaraman, and J.~D. Ullman, \emph{Mining of Massive Datasets,
  2nd Ed}.\hskip 1em plus 0.5em minus 0.4em\relax Cambridge University Press,
  2014.

\bibitem{DBLP:journals/corr/AfratiDK0U15a}
\BIBentryALTinterwordspacing
F.~N. Afrati, S.~Dolev, E.~Korach, S.~Sharma, and J.~D. Ullman, ``Assignment
  problems of different-sized inputs in {MapReduce},'' \emph{CoRR}, vol.
  abs/1507.04461, 2015. [Online]. Available:
  \url{http://arxiv.org/abs/1507.04461}
\BIBentrySTDinterwordspacing

\bibitem{heintzend}
B.~Heintz, A.~Chandra, R.~Sitaraman, and J.~Weissman, ``{End-to-end
  Optimization for Geo-Distributed MapReduce},'' \emph{IEEE Transactions on
  Cloud Computing}, no.~1, pp. 1--1.

\bibitem{ghadoop}
\BIBentryALTinterwordspacing
L.~Wang, J.~Tao, R.~Ranjan, H.~Marten, A.~Streit, J.~Chen, and D.~Chen,
  ``\textmd{G-Hadoop: MapReduce across distributed data centers for
  data-intensive computing},'' \emph{Future Generation Comp. Syst.}, vol.~29,
  no.~3, pp. 739--750, 2013. [Online]. Available:
  \url{http://dx.doi.org/10.1016/j.future.2012.09.001}
\BIBentrySTDinterwordspacing

\bibitem{geo-hadoop}
S.~Dolev, P.~Florissi, E.~Gudes, S.~Sharma, and I.~Singer, ``A survey on
  geographically distributed data processing using {MapReduce},'' 2016, {Dept.
  of Computer Science, Ben-Gurion University of the Negev, Israel}.

\bibitem{hmr}
\BIBentryALTinterwordspacing
Y.~Luo and B.~Plale, ``Hierarchical {MapReduce} programming model and
  scheduling algorithms,'' in \emph{12th {IEEE/ACM} International Symposium on
  Cluster, Cloud and Grid Computing, CCGrid 2012, Ottawa, Canada, May 13-16,
  2012}, 2012, pp. 769--774. [Online]. Available:
  \url{http://dx.doi.org/10.1109/CCGrid.2012.132}
\BIBentrySTDinterwordspacing

\bibitem{Palanisamy-Purlieus-2011}
\BIBentryALTinterwordspacing
B.~Palanisamy, A.~Singh, L.~Liu, and B.~Jain, ``Purlieus: Locality-aware
  resource allocation for {MapReduce} in a cloud,'' in \emph{Proceedings of
  International Conference for High Performance Computing, Networking, Storage
  and Analysis}, 2011, pp. 58:1--58:11. [Online]. Available:
  \url{http://doi.acm.org/10.1145/2063384.2063462}
\BIBentrySTDinterwordspacing

\bibitem{DBLP:conf/hpdc/ParkLKHM12}
\BIBentryALTinterwordspacing
J.~Park, D.~Lee, B.~Kim, J.~Huh, and S.~Maeng, ``Locality-aware dynamic {VM}
  reconfiguration on {MapReduce} clouds,'' in \emph{The 21st International
  Symposium on High-Performance Parallel and Distributed Computing, HPDC'12,
  Delft, Netherlands - June 18 - 22, 2012}, 2012, pp. 27--36. [Online].
  Available: \url{http://doi.acm.org/10.1145/2287076.2287082}
\BIBentrySTDinterwordspacing

\bibitem{DBLP:journals/corr/abs-1206-4377}
\BIBentryALTinterwordspacing
F.~N. Afrati, A.~D. Sarma, S.~Salihoglu, and J.~D. Ullman, ``Upper and lower
  bounds on the cost of a {Map-Reduce} computation,'' \emph{{PVLDB}}, vol.~6,
  no.~4, pp. 277--288, 2013. [Online]. Available:
  \url{http://www.vldb.org/pvldb/vol6/p277-dassarma.pdf}
\BIBentrySTDinterwordspacing

\bibitem{ding2013commapreduce}
\BIBentryALTinterwordspacing
L.~Ding, G.~Wang, J.~Xin, X.~S. Wang, S.~Huang, and R.~Zhang, ``{ComMapReduce}:
  An improvement of {MapReduce} with lightweight communication mechanisms,''
  \emph{Data Knowl. Eng.}, vol.~88, pp. 224--247, 2013. [Online]. Available:
  \url{http://dx.doi.org/10.1016/j.datak.2013.04.004}
\BIBentrySTDinterwordspacing

\bibitem{TCS}
\BIBentryALTinterwordspacing
P.~Malhotra, P.~Agarwal, and G.~Shroff, ``Graph-parallel entity resolution
  using {LSH} {\&} {IMM},'' in \emph{Proceedings of the Workshops of the
  {EDBT/ICDT} 2014 Joint Conference {(EDBT/ICDT} 2014), Athens, Greece, March
  28, 2014.}, 2014, pp. 41--49. [Online]. Available:
  \url{http://ceur-ws.org/Vol-1133/paper-07.pdf}
\BIBentrySTDinterwordspacing

\bibitem{zaharia2010spark}
\BIBentryALTinterwordspacing
M.~Zaharia, M.~Chowdhury, M.~J. Franklin, S.~Shenker, and I.~Stoica, ``Spark:
  Cluster computing with working sets,'' in \emph{2nd {USENIX} Workshop on Hot
  Topics in Cloud Computing, HotCloud'10, Boston, MA, USA, June 22, 2010},
  2010. [Online]. Available:
  \url{https://www.usenix.org/conference/hotcloud-10/spark-cluster-computing-working-sets}
\BIBentrySTDinterwordspacing

\bibitem{DBLP:conf/sigmod/MalewiczABDHLC10}
\BIBentryALTinterwordspacing
G.~Malewicz, M.~H. Austern, A.~J.~C. Bik, J.~C. Dehnert, I.~Horn, N.~Leiser,
  and G.~Czajkowski, ``Pregel: a system for large-scale graph processing,'' in
  \emph{Proceedings of the {ACM} {SIGMOD} International Conference on
  Management of Data, {SIGMOD} 2010, Indianapolis, Indiana, USA, June 6-10,
  2010}, 2010, pp. 135--146. [Online]. Available:
  \url{http://doi.acm.org/10.1145/1807167.1807184}
\BIBentrySTDinterwordspacing

\bibitem{DBLP:journals/tkde/AfratiU11}
\BIBentryALTinterwordspacing
F.~N. Afrati and J.~D. Ullman, ``Optimizing multiway joins in a {Map-Reduce}
  environment,'' \emph{{IEEE} Trans. Knowl. Data Eng.}, vol.~23, no.~9, pp.
  1282--1298, 2011. [Online]. Available:
  \url{http://dx.doi.org/10.1109/TKDE.2011.47}
\BIBentrySTDinterwordspacing

\bibitem{DBLP:conf/edbt/ZhangLJ12}
\BIBentryALTinterwordspacing
C.~Zhang, F.~Li, and J.~Jestes, ``Efficient parallel {kNN} joins for large data
  in {MapReduce},'' in \emph{15th International Conference on Extending
  Database Technology, {EDBT} '12, Berlin, Germany, March 27-30, 2012,
  Proceedings}, 2012, pp. 38--49. [Online]. Available:
  \url{http://doi.acm.org/10.1145/2247596.2247602}
\BIBentrySTDinterwordspacing

\bibitem{DBLP:journals/pvldb/LuSCO12}
\BIBentryALTinterwordspacing
W.~Lu, Y.~Shen, S.~Chen, and B.~C. Ooi, ``Efficient processing of k nearest
  neighbor joins using mapreduce,'' \emph{{PVLDB}}, vol.~5, no.~10, pp.
  1016--1027, 2012. [Online]. Available:
  \url{http://vldb.org/pvldb/vol5/p1016_weilu_vldb2012.pdf}
\BIBentrySTDinterwordspacing

\end{thebibliography}
%\end{spacing}


\ifCLASSOPTIONcaptionsoff
  \newpage
\fi

\end{document}